\documentclass[11pt]{article}
\usepackage{thmtools}
\usepackage{etex}
\usepackage{geometry}
\usepackage{amssymb,amsmath,amscd}
\usepackage{amsthm}
\usepackage{mathtools}
\usepackage{dsfont}
\usepackage[numbers]{natbib}
\usepackage{booktabs}
\usepackage{url}
\usepackage[toc,page]{appendix}
\usepackage[bb=boondox]{mathalfa}
\usepackage{fancybox}

\usepackage[utf8]{inputenc}
\usepackage{cmbright}
\usepackage[T1]{fontenc}
\usepackage{lmodern}

\usepackage[usenames,dvipsnames,svgnames,table]{xcolor}
\usepackage{tikz}
\usetikzlibrary{arrows,shapes,positioning}
\usepackage{tikz-cd}
\usepackage[font={small,it}]{caption}
\usepackage{multirow}
\usepackage[colorinlistoftodos]{todonotes}

\newcommand*\circled[1]{\tikz[baseline=(char.base)]{
                \node[shape=circle,draw,inner sep=1.5pt] (char) {#1};}}
\newcommand*\mybox[1]{\tikz[baseline=(char.base)]{
                \node[shape=rectangle,draw,inner sep=1.5pt] (char) {$\substack{#1}$};}}

\usepackage{float}

\usepackage{listings}
\usepackage{subcaption}
\usepackage{enumerate}

\usepackage{authblk}

\usepackage{hyperref}
\usepackage{cleveref}

\newtheorem{proposition}{Proposition}
\crefname{proposition}{Proposition}{Propositions}
\Crefname{proposition}{Proposition}{Propositions}

\crefname{theorem}{Theorem}{Theorems}
\Crefname{theorem}{Theorem}{Theorems}
\newtheorem{lemma}{Lemma}
\crefname{lemma}{Lemma}{Lemmas}
\Crefname{lemma}{Lemma}{Lemmas}

\crefname{corollary}{Corollary}{Corollaries}
\Crefname{corollary}{Corollary}{Corollaries}

\theoremstyle{definition}
\crefname{definition}{Definition}{Definitions}
\Crefname{definition}{Definition}{Definitions}
\newtheorem{definition}{Definition}
\newtheorem{question}{Question}
\crefname{question}{Question}{Questions}
\Crefname{question}{Question}{Questions}
\newtheorem{condition}{Condition}
\crefname{condition}{Condition}{Conditions}
\Crefname{condition}{Condition}{Conditions}
\crefformat{equation}{\textup{#2(#1)#3}}

\newtheoremstyle{exampstyle}
  {\topsep} 
  {\topsep} 
  {} 
  {} 
  {\bfseries\itshape} 
  {.} 
  {.5em} 
  {} 

\theoremstyle{exampstyle}

\newcommand{\B}{\mathbb{B}}
\newcommand{\0}{\mathbf{0}}
\newcommand{\C}{\mathcal{C}}
\newcommand{\F}{\mathcal{F}}
\newcommand{\Q}{\mathcal{Q}}

\title{Local negative circuits and cyclic attractors\\in Boolean networks with at most five components}
\author[1]{Elisa Tonello}
\author[1]{Etienne Farcot}
\author[2]{Claudine Chaouiya}
\affil[1]{School of Mathematical Sciences, University of Nottingham, Nottingham, NG7 2RD, UK}
\affil[2]{Instituto Gulbenkian de Ciência, Oeiras, Portugal}
\date{}

\begin{document}

\maketitle

\begin{abstract}
  We consider the following question on the relationship between
  the asymptotic behaviours of asynchronous dynamics of Boolean networks and their regulatory structures:
  does the presence of a cyclic attractor imply the existence of a local negative circuit in the regulatory graph?
  When the number of model components $n$ verifies $n \geq 6$, the answer is known to be negative.
  We show that the question can be translated into a Boolean
  satisfiability problem on $n \cdot 2^n$ variables.
  A Boolean formula expressing the absence of local negative circuits
  and a necessary condition for the existence of cyclic attractors
  is found unsatisfiable for $n \leq 5$.
  In other words, for Boolean networks with up to $5$ components, the presence
  of a cyclic attractor requires the existence of a local negative circuit.
\end{abstract}

\section{Introduction}

Boolean networks are used to model the dynamics resulting from the interactions between $n$ regulatory components
that can assume only two values, $0$ and $1$, and are therefore
naturally described as maps from $\{0,1\}^n$ to itself.
Any such map uniquely identifies an \emph{asynchronous dynamics},
which requires at most one component to change at each step.
A regulatory graph defined by a Boolean network is a graph with one node for each regulatory component,
and directed, signed edges that represent regulatory interactions.
A regulation from a component to another might be observable only at certain states.
Therefore, for each state of the system, a \emph{local} regulatory graph is defined by considering
only the regulations that can be observed at that state.

Since the explicit construction and analysis of asynchronous dynamics is generally impractical,
the capability of regulatory structures to inform about the network dynamics has been often investigated.
In particular, relationships have been established between the presence of circuits in regulatory graphs
and the asymptotic asynchronous behaviours of Boolean networks.
In absence of regulatory circuits, the dynamics always reaches a unique fixed point~\cite{shih2005combinatorial},
whereas local positive circuits are required for multistationarity~\cite{remy2008graphic,richard2007necessary}
and negative circuits for oscillations~\cite{remy2008graphic,richard2010negative}.
Here we consider the following question (for studies addressing related issues, see for example~\cite{remy2008graphic,richard2010negative,richard2011local,richard2013kernels,ruet2017negative}):
\begin{question}\label{q1}
  Does the presence of a cyclic attractor imply the existence of a
  negative circuit in a local regulatory graph?
\end{question}

A counterexample for the multilevel case, i.e., where the discrete variables can take their
values in a broader range than $\{0,1\}$, was presented by Richard~\cite{richard2010negative}.
Recently, a number of counterexamples have been identified for the Boolean setting.
Ruet~\cite{ruet2017negative} exhibited a procedure to create counterexamples in the Boolean case, for
every $n\geq 7$, $n$ being the number of variables; these are maps admitting an antipodal attractive cycle
and no local negative circuits in the regulatory graph.
Tonello~\cite{tonello2017ontheconversion} and Faur{\'e} and Kaji~\cite{faure2018circuit}
identified different Boolean versions of Richard's discrete example,
that provide counterexamples to~\cref{q1} for $n=6$.
A map with an antipodal attractive cycle and no local regulatory circuits
also exists for $n=6$ (we present such a map in~\cref{appendix:antip6}).

\cref{q1} remains open for $n \leq 5$.
Even for such a small number of components, the range of possible dynamical behaviours is vast,
and connections between the network regulatory structure and its associated dynamics are not immediate.
However, answers to problems such as the one described in \cref{q1} clarify general rules and
can provide guidance, for instance, to gene network modellers seeking to capture a certain dynamical behaviour.

In this work, we describe how \cref{q1} can be translated into a Boolean satisfiability problem (SAT).
To this end, for a fixed number $n$ of regulatory components,
we consider $n \cdot 2^n$ Boolean variables, representing the values taken by the $n$ components
of the Boolean map on the $2^n$ states in $\{0,1\}^n$.
We then describe how the features referred to in~\cref{q1}
can be encoded as Boolean expressions on the $n \cdot 2^n$ variables.
More precisely, we define a Boolean formula that encodes both the absence of local negative circuits
and a necessary condition for the presence of a cyclic attractor.
In addition, we reduce the search space by exploiting symmetries of regulatory networks,
so that, for small $n$, the problem can be analysed by a satisfiability solver in a few hours.
The solver finds the formula unsatisfiable for $n \leq 5$,
and provides further examples for $n=6$.

The relevant definitions and background are introduced in~\cref{sec:background}, whereas~\cref{sec:bsp}
is dedicated to recasting~\cref{q1} as a Boolean satisfiability problem.
We discuss our results in~\cref{sec:conclusion}.

\section{Background}\label{sec:background}
In this section, we fix some notations and introduce the main definitions.
We denote by $\B$ the set $\{0,1\}$, and consider $n \in \mathbb{N}$.
The elements of $\B^n$ are also called \emph{states}.
The state $x \in \B^n$ with $x_i=0$, $i=1,\dots,n$ will be denoted $\0$.
Given $x \in \B^n$ and a set of indices $I \subseteq \{1,\dots,n\}$,
we denote by $\bar{x}^I$ the state that satisfies
$\bar{x}^I_i=1-x_i$ for $i \in I$, and $\bar{x}^I_i=x_i$ for $i\notin I$.
If $I=\{i\}$ for some $i$, we simply write $\bar{x}^i$ for $\bar{x}^{\{i\}}$.
Given two states $x,y \in \B^n$, $d(x,y)$ denotes the Hamming distance between $x$ and $y$.
We call $n$-dimensional \emph{hypercube graph} the directed graph on $\B^n$
with an edge from $x \in \B^n$ to $y \in \B^n$ whenever $d(x,y)=1$.

A Boolean network is defined by a map $f\colon\B^n \rightarrow \B^n$.
The dynamical system defined by $f$ is also referred to as the \emph{synchronous dynamics}.
The \emph{asynchronous state transition graph} or \emph{asynchronous dynamics} $AD_f$
defined by $f$ is a graph on $\B^n$ with an edge from $x \in \B^n$ to $\bar{x}^i$
for all $i \in \{1,\dots,n\}$ such that $f_i(x)\neq x_i$.
We write $(x,y)$ for the edge (transition) from $x$ to $y$.

A non-empty subset $D \subseteq \B^n$ is \emph{trap domain} for $AD_f$ if,
for every edge $(x, y)$, $x \in D$ implies $y \in D$.
The minimal trap domains with respect to the inclusion are called \emph{attractors}
for the dynamics of the network.
Attractors that consist of a single state are called \emph{fixed points} or \emph{stable states};
the other attractors are referred to as \emph{cyclic attractors}.

Boolean networks are used to model the interactions between regulatory components. 
The interactions are derived from a Boolean map $f$ as follows.
For each state $x \in \B^n$, we define the \emph{local regulatory graph} $G_f(x)$ of $f$ at $x \in \B^n$
as a labelled directed graph with $\{1,\dots,n\}$ as set of nodes.
The graph $G_f(x)$ contains an edge from node $j$ to node $i$, also called \emph{interaction}
between $j$ and $i$, when $f_i(\bar{x}^j)\neq f_i(x)$;
the edge is represented as $j \rightarrow i$ and is labelled with $s=(\bar{x}^j_j-x_j) \cdot (f_i(\bar{x}^j)-f_i(x))$.
The label $s$ is also called the \emph{sign} of the interaction,
and accounts for the regulatory effect of $j$ upon $i$ at the state $x$.

The \emph{global regulatory graph} $G_f$ of $f$ is the multi-directed labelled graph on $\{1,\dots,n\}$
that contains an edge $j\rightarrow i$ of sign $s$ if
there exists a state for which the local regulatory graph contains an interaction $j\rightarrow i$ of sign $s$.
In the global regulatory graph parallel edges are permitted to account for different
regulatory effects that can be observed at different states.

The sign of a path $i_1\rightarrow i_2\rightarrow \dots \rightarrow i_k$ in a regulatory graph is 
defined as the product of the signs of its edges.
A \emph{circuit} in a regulatory graph is a path $i_1\rightarrow i_2\rightarrow \dots \rightarrow i_k$
with $i_1=i_k$ and such that the indices $i_1,\dots,i_{k-1}$ are all distinct.
We recall a useful result which can be found in~\cite[Remark 1]{richard2011local} and~\cite[Lemma 5.2]{ruet2016local}.

\begin{lemma}\label{lemma:sign_circuit_parity}
  Let $C$ be a circuit of $G_f(x)$ with set of vertices $I$.
  If the cardinality of $\{i \in I | \ f_i(x) \neq x_i\}$ is even (resp., odd),
  then $C$ is a positive (resp.\ negative) circuit.
\end{lemma}

\subsection{Regulatory circuits and asymptotic behaviours}
Following R. Thomas early conjectures~\cite{thomas1981relation}, asymptotic properties of the asynchronous state transition graph
have been connected to the existence and the signs of regulatory circuits.

Shih and Dong~\cite{shih2005combinatorial} established that, if no local regulatory circuit exists,
then the map admits a unique fixed point.
The result was extended to the multilevel setting by Richard~\cite{richard2008extension}.

The presence of multiple attractors was shown to require the existence of a
local positive circuit~\cite{richard2007necessary}.
The existence of a cyclic attractor requires instead the (global) regulatory graph to
include a negative circuit. This was proved in~\cite{remy2008graphic} for the case of an
attractive cycle (a cycle in the asynchronous dynamics that is an attractor),
and in the general case of a cyclic attractor in~\cite{richard2010negative}.

Cyclic attractors are compatible, however, with the absence of local negative circuits.
This was first shown in~\cite{richard2010negative} in the multilevel case.
Boolean networks with a cyclic attractor and no local negative circuits
were presented in~\cite{ruet2017negative}, with a method to create maps
with antipodal attractive cycles and no local negative circuits, for $n\geq 7$.
Tonello~\cite{tonello2017ontheconversion} and Faur{\'e} and Kaji~\cite{faure2018circuit}
exhibited maps with cyclic attractors and no local negative circuits, for $n=6$.
Maps with antipodal attractive cycles and no local negative circuits
also exist for $n=6$; a procedure that extends the one
in~\cite{ruet2017negative} is presented, for completeness, in~\cref{appendix:antip6}.

In this work we consider~\cref{q1} in the remaining cases ($n\leq 5$).
We show that the problem can be approached as a Boolean satisfiability problem,
and find that all maps from $\B^n$ to itself with a cyclic attractor define a local negative circuit.

\subsection{Automorphisms of the $n$-hypercube}
In this section we present some relationships between Boolean networks and symmetries of the hypercube;
these will be used to translate~\cref{q1} into a Boolean expression
(see~\cref{sec:cyclic_attr}).

We first introduce some additional notations. Given $I \subseteq \{1,\dots,n\}$,
$\psi_I$ denotes the map defined by $\psi_I(x) = \bar{x}^I$ for all $x \in \B^n$.
We call $S_n$ the group of permutations of $\{1,\dots,n\}$;
$S_n$ acts on $\B^n$ by permuting the coordinates: for $\sigma \in S_n$, $\sigma(x)=(x_{\sigma^{-1}(1)},\dots,x_{\sigma^{-1}(n)})$.
We consider here the maps of the form $U = \psi_I \circ \sigma$
for some $I \subseteq \{1,\dots,n\}$ and some $\sigma \in S_n$.
These are all the automorphisms of the $n$-hypercube (see for instance~\cite{slepian1953number,ruet2017negative}).

Given the maps $U =\psi_I \circ \sigma$ and $f:\B^n \rightarrow \B^n$,
we write $f^U = U \circ f \circ U^{-1}$.
The following proposition relates the asynchronous state transition graphs and regulatory graphs of $f$ and $f^U$,
asserting that they have the same structures.
In addition, albeit the signs
of the interactions of the regulatory graphs can differ, the signs of the regulatory circuits are the same.
An example illustrating this property is given in~\cref{fig:sym_2}.

\begin{proposition}\label{prop:change_circuits}
  Consider the maps $U = \psi_I \circ \sigma$ and $f: \B^n \rightarrow \B^n$. 

  \begin{itemize}
    \item[(i)] The state transition graphs $AD_f$ and $AD_{f^U}$ are isomorphic.
    \item[(ii)] For each $x \in \B^n$, the graphs $G_f(x)$ and $G_{f^U}(U(x))$,
      seen as unlabelled directed graphs, are isomorphic. In addition, corresponding circuits have the same signs.
  \end{itemize}

\end{proposition}
\begin{proof}
  $(i)$ We have that $(x, \bar{x}^i)$ is in $AD_f$ if and only if $(U(x), U(\bar{x}^i)=\overline{U(x)}^{\sigma(i)})$
  is in $AD_{f^U}$, so that the graph isomorphism is given by $U$.
  This follows from the observation that
  \begin{equation}\label{eq:fUx}
      f^U_{\sigma(i)}(U(x))=\overline{\sigma(f(x))}^I_{\sigma(i)}=\overline{f(x)}_i^{\sigma^{-1}(I)},
  \end{equation}
  and $U(x)_{\sigma(i)}=\overline{\sigma(x)}^I_{\sigma(i)}=\overline{x}_i^{\sigma^{-1}(I)}$, and therefore
  $f^U_{\sigma(i)}(U(x))\neq U(x)_{\sigma(i)}$ if and only if
  $f_i(x)\neq x_i$.

  $(ii)$ The graph $G_{f^U}(U(x))$ contains an interaction $\sigma(j) \to \sigma(i)$
  if and only if $f^U$ verifies $f^U_{\sigma(i)}(\overline{U(x)}^{\sigma(j)}) \neq f^U_{\sigma(i)}(U(x))$.
  Since $\overline{U(x)}^{\sigma(j)}=U(\bar{x}^j)$, as a consequence of~\cref{eq:fUx} we have that 
  $f^U_{\sigma(i)}(\overline{U(x)}^{\sigma(j)}) = \overline{f(\bar{x}^j)}_i^{\sigma^{-1}(I)}$,
  hence the graph $G_{f^U}(U(x))$ contains the interaction $\sigma(j) \rightarrow \sigma(i)$
  if and only if $f_i(\bar{x}^j) \neq f_i(x)$, i.e.\ if and only if
  $j \rightarrow i$ is an interaction in $G_{f}(x)$.

  Given a circuit $C$ in $G_f(x)$ with support on $L \subseteq \{1,\dots,n\}$,
  $\sigma(L)$ is therefore the support of a circuit $C^U$ in $G_{f^U}(U(x))$.
  In addition, from point (i), we have that the sets $\{i \in L| f_i(x)\neq x_i\}$ and
  $\{i \in \sigma(L)| f^U_i(U(x))\neq U(x)_i\}$ have the same cardinality.
  We conclude by observing that, by~\cref{lemma:sign_circuit_parity}, the circuit $C$ is positive (resp. negative)
  if and only if the cardinality of $\{i \in L| f_i(x)\neq x_i\}$ is even (resp. odd),
  hence if and only if the cardinality of $\{i \in \sigma(L)| f^U_i(U(x))\neq U(x)_i\}$ is
  even (resp. odd), if and only if $C^U$ is positive (resp. negative).
\end{proof}

It follows from the proposition that a property relating the asymptotic behaviour of the asynchronous dynamics and
the regulatory circuits holds for a map if and only if it holds for any of its conjugated maps under symmetry.
We will use this fact when writing~\cref{q1} as a Boolean satisfiability problem in the next section.

\begin{figure}[t]
  \centering
  \begin{subfigure}[b]{7cm}
  \begin{tikzcd}[ampersand replacement=\&,column sep=scriptsize]
    01 \arrow[r] \arrow[d] \& 11 \arrow[d, rightharpoonup, xshift=0.5pt] \\
    00 \& 10 \arrow[l] \arrow[u, rightharpoonup, xshift=-0.5pt]
  \end{tikzcd}
  \quad
  \begin{tikzcd}[ampersand replacement=\&,column sep=scriptsize]
    \circled{1} \arrow[r, bend left=30] \& \circled{2} \arrow[l, bend left=30] \arrow[loop right,-|]
  \end{tikzcd}
  \subcaption{\label{fig:sym_2_b}}
  \end{subfigure}
  \begin{subfigure}[b]{7cm}
  \begin{tikzcd}[ampersand replacement=\&,column sep=scriptsize]
    01 \& 11 \arrow[d] \arrow[l] \\
    00 \arrow[r, rightharpoonup, yshift=0.5pt] \arrow[u] \& 10 \arrow[l, rightharpoonup, yshift=-0.5pt]
  \end{tikzcd}
  \quad
  \begin{tikzcd}[ampersand replacement=\&,column sep=scriptsize]
    \circled{1} \arrow[loop left,-|] \arrow[r, bend left=30,-|] \& \circled{2} \arrow[l, bend left=30,-|]
  \end{tikzcd}
  \subcaption{\label{fig:sym_2_a}}
  \end{subfigure}
  \caption{The graphs in (a) and (b) represent the asynchronous state transition graphs and the regulatory graphs of the
    maps $f: (x_1,x_2) \mapsto (x_2,x_1(1-x_2))$ and $g: (x_1,x_2) \mapsto ((1-x_1)(1-x_2),1-x_1)$ respectively.
    Standard arrows $j \rightarrow i$ denote interactions with positive sign, and arrows with a vertical tip $j \dashv i$ represent negative interactions.
    The asynchronous state transition graphs have the same ``shape'': the map in $(b)$ can be obtained from the map in $(a)$
    by swapping the two components, and changing $0$ with $1$ for the second component.
    In other words, $g = U \circ f \circ U^{-1}$, with $U=\psi_I\circ\sigma$, $\psi_I\colon(x_1,x_2)\mapsto(x_1,1-x_2)$ and $\sigma\colon(x_1,x_2)\mapsto(x_2,x_1)$.
    The regulatory graphs of the two maps also have the same edges. The positive interactions on the left correspond to negative interactions 
    on the right; however, the sign of the loop is negative in both regulatory graphs, and the sign of the circuit
    involving the two components is positive in both graphs.}\label{fig:sym_2}
\end{figure}
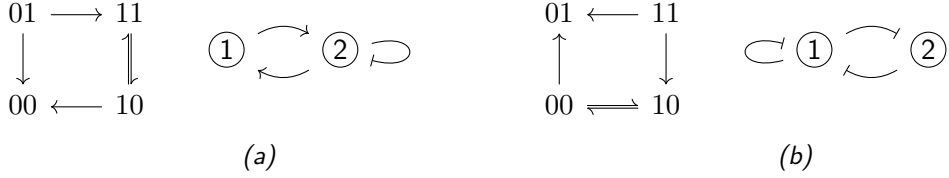

\section{Recasting \cref{q1} as a Boolean satisfiability problem}\label{sec:bsp}
For each $n$, \cref{q1} requires that we determine (or exclude the existence of) a map $f$ from $\B^n = {\{0,1\}}^n$ to itself.
We therefore consider as variables of the problem $n \cdot 2^n$ Boolean
variables that we denote as
\begin{equation}\label{var:bool}
  f_1(x), \dots, f_n(x), \ x \in \B^n.
\end{equation}
We first describe how the absence of negative circuits in the
local regulatory graph $G_f(x)$ can be translated into a set of expressions on
the variables~\cref{var:bool}.

\subsection{Imposing the absence of local negative circuits}\label{sec:no_lnc}

To express the sign condition on the circuits, we consider each local graph as a
complete graph on the nodes $\{1, \dots, n\}$. Then, we consider every possible
circuit on this graph, and we impose that each circuit has a non-negative sign.
For small values of $n$, this requirement leads to a satisfiability problem
that is computationally manageable.
The number of elementary circuits of length $k$ in a complete graph on $n$ nodes
is given by ${n \choose k} (k-1)!$. Hence we have to consider, for instance,
$89$ circuits for $n=5$, and $415$ circuits for $n=6$.
Let $\C_n$ denote the set of all possible circuits on the complete graph on $\{1,\dots,n\}$.

Given a state $x \in \B^n$, if an interaction exists in $G_f(x)$ from $j$ to $i$, then its sign is given by the difference
$f_i(x_1, \dots, x_{j-1}, 1, x_{j+1}, \dots, x_n)-f_i(x_1, \dots, x_{j-1}, 0, x_{j+1}, \dots, x_n)$.
We define
\[l_x^0(j,i) = f_i(x_1, \dots, x_{j-1}, 0, x_{j+1}, \dots, x_n), \hspace{10pt} l_x^1(j,i) = f_i(x_1, \dots, x_{j-1}, 1, x_{j+1}, \dots, x_n).\]
The following Boolean expression asserts that the interaction from $j$ to $i$ is positive:
\[\mathcal{P}^x(j,i) = l^1_x(j,i) \wedge \neg l^0_x(j,i),\]
and the following Boolean expression asserts that the interaction is negative:
\[\mathcal{N}^x(j,i) = \neg l^1_x(j,i) \wedge l_x^0(j,i).\]

We can now write a formula expressing that, given a state $x$, a circuit $c$ is negative in $G_f(x)$,
that is to say, the circuit $c$ contains an odd number of negative interactions, the remaining interactions being positive.
We write $m$ for the length of the circuit, and $c^-$ and $c^+$ for the interactions in $c$ with negative or positive sign, respectively.
We obtain the following formula:
\begin{equation}\label{eq:neg_local}
\Phi^x_c = \bigvee_{\substack{1 \leq k \leq m,\ k\ \text{odd} , \\ c = c^- \cup c^+,\ \#c^- = k}} \left(\bigwedge_{j \rightarrow i \ \text{in}\  c^-} \mathcal{N}^x(j,i) \wedge \bigwedge_{j \rightarrow i \ \text{in}\  c^+} \mathcal{P}^x(j,i) \right).
\end{equation}
The absence of local negative circuits in the regulatory graph is therefore specified by the formula
\begin{equation}\label{eq:no_neg_local}
  \begin{aligned}
\bigwedge_{x \in \B^n, c \in \C_n} \neg \Phi^x_c
& = \bigwedge_{x \in \B^n, c \in \C_n} \neg \left(\bigvee_{\substack{1 \leq k \leq m,\ k\ \text{odd} , \\ c = c^- \cup c^+,\ \#c^- = k}}
  \left(\bigwedge_{j \rightarrow i \ \text{in}\  c^-} \mathcal{N}^x(j,i) \wedge \bigwedge_{j \rightarrow i \ \text{in}\  c^+} \mathcal{P}^x(j,i)\right)\right), \\
  \end{aligned}
\end{equation}
which we can write in CNF form as
\begin{equation*}
  \begin{aligned}
\bigwedge_{\substack{x \in \B^n\\c \in \C_n}} \neg \Phi^x_c
& = \bigwedge_{\substack{x \in \B^n\\c \in \C_n}} \bigwedge_{\substack{1 \leq k \leq m,\ k\ \text{odd} , \\ c = c^- \cup c^+,\ \#c^- = k}}
  \left(\bigvee_{j \rightarrow i \ \text{in}\  c^-} l^1_x(j,i) \vee \neg l_x^0(j,i) \vee \bigvee_{j \rightarrow i \ \text{in}\  c^+} \neg l^1_x(j,i) \vee l^0_x(j,i)\right).
  \end{aligned}
\end{equation*}

\subsection{A simpler question: absence of fixed points}
Before considering~\cref{q1} in its generality, we describe how a special case of the question
can be easily translated into a Boolean satisfiability problem.
The question is the following:
\begin{question}\label{q2}
  Does the absence of fixed points imply the existence of a local
  negative circuit in the regulatory graph?
\end{question}
The absence of local negative circuits being formulated as in~\cref{sec:no_lnc},
we now need to formulate the absence of fixed points.
To express that a state $x \in \B^n$ is not a fixed point for $f$ we can write
the following formula:
\begin{equation}\label{eq:fp}
  \F^x = \bigvee_{\substack{1\leq i \leq n\\x_i=0}} f_i(x) \lor \bigvee_{\substack{1\leq i \leq n\\x_i=1}} \neg f_i(x).
\end{equation}
The formula expressing the absence of fixed points for $f$ can be written as:
\begin{equation}\label{eq:no_fixed_points}
\bigwedge_{x \in \B^n} \F^x = \bigwedge_{x \in \B^n}\left(\bigvee_{\substack{1\leq i \leq n\\x_i=0}} f_i(x) \lor \bigvee_{\substack{1\leq i \leq n\\x_i=1}} \neg f_i(x)\right).
\end{equation}
Since the state $\0$ is not fixed, there exists an index $i$ such that $f_i(\0)=1$.
Consider a permutation $\sigma \in S_n$ that sends $i$ to $1$. The map $g = \sigma \circ f \circ \sigma^{-1}$
satisfies $g_1(\0)=1$; in addition, by \cref{prop:change_circuits}, $g$ and $f$ have local circuits
with the same signs. We can therefore assume that the first coordinate of $f(\0)$ is $1$.
The formula corresponding to \cref{q2} is therefore:
\begin{equation}\label{eq:q2}
\Q_2 = \left(\bigwedge_{x \in \B^n} \F^x\right) \wedge
       \left(\bigwedge_{x \in \B^n, c \in \C_n} \neg \Phi^x_c\right) \wedge f_1(\0).
\end{equation}
The unsatisfiability of this problem is thus determined, for $n=5$, in minutes, by the satisfiability solvers we considered (see \cref{sec:results}).
The solvers also identify other examples of maps with no fixed points and no local negative circuits in the regulatory graph, for $n=6$.
The existence of a cyclic attractor is less straightforward to express; we describe our approach in the next section.

\subsection{A necessary condition for the existence of a cyclic attractor}\label{sec:cyclic_attr}
In this section we consider~\cref{q1} in its generality.
We need therefore to assert that the asynchronous state transition graph of $f$ admits
a cyclic attractor.
The approach is based on the following observation.

\begin{proposition}\label{prop:paths}
  The asynchronous state transition graph $AD_f$ of a map $f: \B^n \rightarrow \B^n$ admits a
  cyclic attractor if and only if there exists a state $x \in \B^n$
  such that, for any $y \in \B^n$, if there is a path in $AD_f$ from $x$ to $y$, then $y$ is not a fixed point.
\end{proposition}
\begin{proof}
  If $AD_f$ admits a cyclic attractor, then the conclusion is true
  for any state $x$ in the cyclic attractor.

  Conversely, suppose that $x$ is a state with the described property, and
  call $R$ the set of points reachable from $x$ in the asynchronous state transition graph.
  Then the minimal trap domain contained in $R$
  does not contain any fixed point, hence it must contain a cyclic attractor for $AD_f$.
\end{proof}

\cref{prop:paths} translates the existence of a cyclic attractor
into a condition on the paths in the asynchronous state transition graph.
It is, however, computationally problematic to impose that, if $AD_f$ contains a
path of \emph{any length} from $x$ to $y$, then $y$ is not a fixed point.
We therefore consider the following condition instead.

\begin{condition}\label{cond:gammak}
  There exists a state $x \in \B^n$ such that, for each $y \in \B^n$,
  if there is an acyclic path in $AD_f$ from $x$ to $y$ of length at most $k$,
  then $y$ is not a fixed point.
\end{condition}

It is clear from \cref{prop:paths} that, for each $k\geq 0$,
\cref{cond:gammak} is a necessary condition for the existence of a cyclic attractor.
Our strategy is therefore to impose the absence of local negative circuits,
as well as \cref{cond:gammak} for increasing values of $k$, until
we find that the problem is unsatisfiable.

In order to express \cref{cond:gammak}, we need to encode
the existence of a given path in the asynchronous state transition graph.
Given a pair of states $(x, y)$ such that $d(x, y)=1$, if $x_j\neq y_j$ we can require that
the edge $(x, y)$ is in $AD_f$ by imposing
\begin{equation}\label{eq:step}
  f_j(x) \text{ if } y_j = 1\text{, else } \neg f_j(x).
\end{equation}
Given a sequence of states $\pi = (x^0, x^1, \dots, x^k)$ such that $d(x^i, x^{i+1})=1$,
$i=0, \dots, k-1$, we can require that the sequence defines a path in $AD_f$
by imposing $k$ constraints of the form in~\cref{eq:step}:
\begin{equation}\label{eq:path}
  \Theta^\pi = \bigwedge_{\substack{0\leq i\leq k-1\\j \text{ s.t. } x^i_j \neq x^{i+1}_j\\x^{i+1}_j=0}}\neg f_j(x^i) \wedge \bigwedge_{\substack{0\leq i\leq k-1\\j \text{ s.t. } x^i_j \neq x^{i+1}_j\\x^{i+1}_j=1}}f_j(x^i).
\end{equation}

Given a state $x \in \B^n$, let $P^k(x)$ denote the set of acyclic paths in the
$n$-dimensional hypercube graph that start from $x$ and have length less or equal to $k$.
If $\pi$ is a path in $AD_f$, we write $t(\pi)$ for the last node of the path.
We express \cref{cond:gammak} for a state $x \in \B^n$, using~\cref{eq:fp}, as follows:
\begin{equation}\label{eq:condgammakx}
  \bigwedge_{\pi \in P^k(x)} \left(\Theta^\pi \Rightarrow \F^{t(\pi)}\right) =
  \bigwedge_{\pi \in P^k(x)} \neg \Theta^\pi \vee \F^{t(\pi)}.
\end{equation}

\Cref{cond:gammak} requires the existence of a state $x \in \B^n$ that verifies~\cref{eq:condgammakx}.
Suppose that a map $f$ satisfies \cref{eq:condgammakx} for some $x \in \B^n$,
and that its local regulatory graphs do not admit any negative circuit.
Consider $j$ such that $f_j(x)\neq x_j$, and consider a permutation $\sigma \in S_n$ that swaps $j$ and $1$.
Define $I = \{i \in \{1,\dots,n\}| \sigma(x)_i\neq 0\}$.
Then, by \cref{prop:change_circuits}, the map $f^U$ with $U=\psi_I \circ \sigma$ satisfies \cref{eq:condgammakx} for $x=\0$,
and its local regulatory graphs do not admit any negative circuit. In addition, $f_1(\0)=1$.
We have therefore that, to exclude the existence of maps with cyclic attractors and no local negative circuits,
it is sufficient to consider expression~\cref{eq:condgammakx} for $x=\0$, and assume $f_1(\0)=1$.
By combining~\cref{eq:condgammakx} with~\cref{eq:no_neg_local}, we have, for fixed $k$,
the Boolean formula
\begin{equation}\label{eq:q1}
\Q_1 = \left(\bigwedge_{\pi \in P^k(\0)} \neg \Theta^\pi \vee \F^{t(\pi)}\right) \wedge
     \left(\bigwedge_{x \in \B^n, c \in \C_n} \neg \Phi^x_c\right) \wedge f_1(\0),
\end{equation}
which we can use to answer \cref{q1}.
Notice that $\Q_1$ is a generalisation of~\cref{eq:q2}, where fewer points are
required to be non-fixed.
Using \cref{eq:condgammakx} and \cref{eq:path},~\cref{eq:q1} is easily written in CNF form.

\subsection{Results}\label{sec:results}
\begin{table}[t]
  \centering
  \begin{tabular}{c|cccc}
    \multirow{2}{*}{n} & absence of & absence of local & \multirow{2}{*}{$k$} & \multirow{2}{*}{\cref{cond:gammak}} \\
                       & fixed points & negative circuits & & \\
    \hline
    2 & 4 & 16 & 2 & 4 \\
    3 & 8 & 136 & 4 & 39 \\
    4 & 16 & 1,536 & 6 & 1,036 \\
    5 & 32 & 23,328 & 11 & 2,595,405
  \end{tabular}
  \caption{Number of clauses generated by the constraints used to answer \cref{q2} and \cref{q1}.
    $k$ is the path length considered for \cref{cond:gammak},
    and is the minimum path length such that, in a Boolean model with $n$ variables, \cref{eq:q1} is unsatisfiable,
    i.e.\ if all paths from state $\0$ of length at most $k$ do not reach a fixed point,
    there must exist a local negative circuit.}\label{table:results}
\end{table}
We created CNF files in DIMACS CNF format, a standard input format accepted by most SAT solvers.
The files start with a line that begins with \texttt{p cnf} followed by the number
of variables and the number of clauses. One line for each clause then follows.
Each clause is expressed by listing the indices of the variables involved in the clause separated by spaces,
using negative numbers for negated variables. A zero is added at the end of each clause line.
The files were created with a Python script (source code available at \href{https://github.com/etonello/regulatory-network-sat}{github.com/etonello/regulatory-network-sat}).

Using the satisfiability solver Lingeling~\cite{biere2016splatz},
we found that, if $k$ is set to $2,4,6,11$ respectively, for $n=2,3,4,5$,
the problem described by~\cref{eq:q1} is unsatisfiable.
This means that, for $n\leq 5$, all maps
that admit a cyclic attractor must have a local negative circuit.

The lengths $k=2,4,6,11$ are the minimum lengths that lead to the unsatisfiability of the formula in~\cref{eq:q1}.
In other words, there exists at least one map in dimension $2$ (respectively $3$, $4$ and $5$) such that
the paths of length at most $1$ (respectively $3$, $5$ and $10$) do not reach a fixed point,
and the associated regulatory graph does not admit a local negative circuit.
Examples of such maps are given in \cref{fig:counterex_smaller_k}, for $n=2$ and $n=3$.
\Cref{fig:idea} illustrates instead the idea of the result obtained for $n=2$ and $n=3$,
for two special cases of asynchronous state transition graphs admitting a unique path leaving the origin:
since this path reaches $3$ (respectively $5$) different states, the regulatory graph must admit a local negative circuit,
somewhere in the state space.

The CNF file for $n=5$ and $k=11$ on the $160$ variables consists of about 2.6 million clauses
(the number of clauses for each constraint is given in \cref{table:results}).
The satisfiability solver Lingeling~\cite{biere2016splatz}
was used to determine the unsatisfiability and to generate a proof, expressed in the standard DRAT notation~\cite{wetzler2014drat}.
For $n=5$ and $k=11$, the file for the proof is about $1$GB in size,
and was verified using the SAT checking tool chain GRAT~\cite{lammich2017thegrat}.
The CNF file and the proof of unsatisfiability generated for $n=5$, $k=11$ are available as Supplementary Materials.

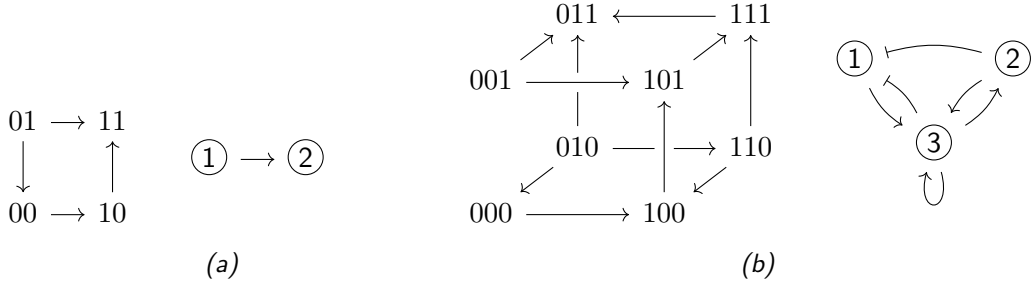
\begin{figure}[t]
  \centering
  \begin{subfigure}[b]{6cm}
  \begin{tikzcd}[ampersand replacement=\&,column sep=small]
    01 \arrow[r] \arrow[d] \& 11 \\
    00 \arrow[r] \& 10 \arrow[u]
  \end{tikzcd}
  \quad
  \begin{tikzcd}[ampersand replacement=\&,column sep=small]
    \circled{1} \arrow[r] \& \circled{2}
  \end{tikzcd}
  \subcaption{\label{fig:counterex_smaller_k_a}}
  \end{subfigure}
  \begin{subfigure}[b]{8cm}
  \begin{tikzcd}[ampersand replacement=\&,column sep=tiny,row sep=small]
    \& 011 \arrow[from=dd] \& \& 111 \arrow[ll] \\
    001 \arrow[ur] \arrow[rr,crossing over] \& \& 101 \arrow[ru] \& \\
    \& 010 \arrow[rr] \arrow[ld] \& \& 110 \arrow[uu] \arrow[ld] \\
    000 \arrow[rr] \& \& 100 \arrow[uu,crossing over] \&
  \end{tikzcd}
  \quad
  \begin{tikzcd}[ampersand replacement=\&,column sep=tiny,row sep=small]
    \circled{1} \arrow[rd,bend right=15] \& \& \circled{2} \arrow[ld,bend right=15] \arrow[ll,-|,bend right=15] \\
    \& \circled{3} \arrow[loop below] \arrow[ru,bend right=15] \arrow[lu,-|,bend right=15] \& \\
  \end{tikzcd}
  \subcaption{\label{fig:counterex_smaller_k_b}}
  \end{subfigure}
  \caption{Example showing that \cref{cond:gammak} is compatible with the absence of local negative circuits for $n=2$ with $k=1$, and for $n=3$ with $k=3$.
    (a) The asynchronous state transition graph and the regulatory graph for the map $f(x_1,x_2)=(1,x_1)$.
    The path of length $2$ leaving the origin reaches a fixed point, and the regulatory graph does not admit any local circuit.
    (b) The asynchronous state transition graph and the (global) regulatory graph for the map $f(x_1,x_2,x_3)=(1-x_2x_3,x_3,x_1x_2x_3-x_1x_2-x_1x_3-x_2x_3+x_1+x_2+x_3)$.
    The path of length $4$ leaving the origin reaches a fixed point; none of the negative circuits admitted by regulatory graph are local.}\label{fig:counterex_smaller_k}
\end{figure}
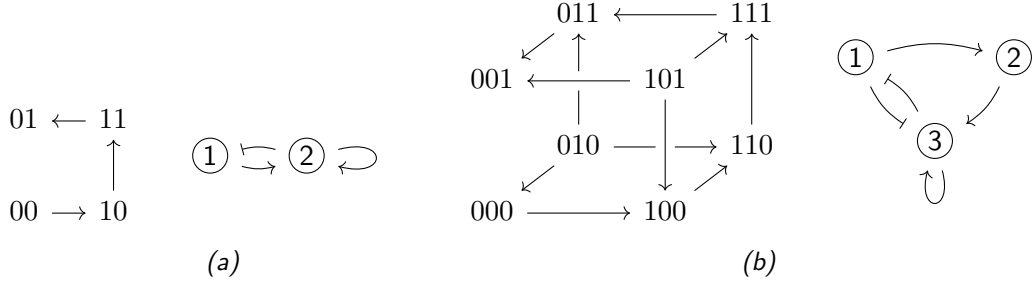
\begin{figure}[t]
  \centering
  \begin{subfigure}[b]{6cm}
  \begin{tikzcd}[ampersand replacement=\&,column sep=small]
    01 \& 11 \arrow[l] \\
    00 \arrow[r] \& 10 \arrow[u]
  \end{tikzcd}
  \quad
  \begin{tikzcd}[ampersand replacement=\&,column sep=small]
    \circled{1} \arrow[r,bend right=15] \& \circled{2} \arrow[l,-|,bend right=15] \arrow[loop right]
  \end{tikzcd}
  \subcaption{\label{fig:idea_a}}
  \end{subfigure}
  \begin{subfigure}[b]{8cm}
  \begin{tikzcd}[ampersand replacement=\&,column sep=tiny,row sep=small]
    \& 011 \arrow[ld] \arrow[from=dd] \arrow[from=dd,to=ddrr] \& \& 111 \arrow[ll] \\
    001 \& \& 101 \arrow[ll,crossing over] \arrow[dd,crossing over] \arrow[ur] \& \\
    \& 010 \arrow[ld] \& \& 110 \arrow[uu] \\
    000 \arrow[rr] \& \& 100 \arrow[ur] \&
  \end{tikzcd}
  \quad
  \begin{tikzcd}[ampersand replacement=\&,column sep=tiny,row sep=small]
    \circled{1} \arrow[rr,bend left=15] \arrow[rd,-|,bend right=15] \& \& \circled{2} \arrow[ld,bend left=15] \\
    \& \circled{3} \arrow[loop below] \arrow[lu,-|,bend right=15] \& \\
  \end{tikzcd}
  \subcaption{\label{fig:idea_b}}
  \end{subfigure}
  \caption{(a) The asynchronous state transition graph and the regulatory graph for the map $f(x_1,x_2)=(1-x_2,x_1+x_2-x_1x_2)$.
    The paths leaving the origin do not reach a fixed point in $2$ steps, hence a local negative circuit must exist in the regulatory graph.
    The unique attractor for the asynchronous state transition graph is a fixed point.
    (b) The asynchronous state transition graph and the (global) regulatory graph for the map $f(x_1,x_2,x_3)=(1-x_3,x_1,x_1x_2x_3-x_1x_3-x_2x_3+x_2+x_3)$.
    No local negative circuit of dimension $1$ or $2$ exists; however, since the only path leaving the origin has length $5$,
    the regulatory graph must admit a local negative circuit involving all three variables.
    The unique attractor for the asynchronous state transition graph is a fixed point.}\label{fig:idea}
\end{figure}

\section{Conclusion}\label{sec:conclusion}
In this work we have considered the question of whether
a regulatory network whose asynchronous state transition graph
contains a cyclic attractor must admit a local negative circuit.
For $n\geq 6$, only the existence of a negative circuit in the global regulatory structure is guaranteed~\cite{richard2010negative}.
We have written the question as a Boolean satisfiability problem,
and SAT solvers found the problem unsatisfiable for $n \leq 5$.
Behaviours of gene regulatory networks have been previously investigated using
SAT (see, for instance~\cite{tamura2009detecting,dubrova2011sat,varela2018stable}).
Here we demonstrated that Boolean satisfiability problems can be utilised not only to examine the behaviour
of a given network, but also to explore the existence of maps with desired properties,
specifically, properties of the associated regulatory structure.

We actually verified that, in absence of local negative circuits,
\cref{cond:gammak}, that is implied by the existence of a cyclic attractor, cannot be satisfied, for $k$ sufficiently large.
\cref{cond:gammak} requires that, for at least one state in the state space,
paths of lengths at most $k$ leaving that state cannot reach a fixed point.
We found that \cref{cond:gammak} with $k=2,4,6,11$
is sufficient for the existence of a local negative circuit in the regulatory graph,
for dimensions $n=2,3,4,5$, respectively.
The absence of local negative circuits is instead compatible with \cref{cond:gammak}
for $k\leq 1,3,5$ and $10$, in dimensions $n=2,3,4,5$, respectively.

It is natural to ask whether a relation can be established between the values identified for $k$
via the satisfiability problems and specific properties of the $n$-hypercube.
Such an understanding could help in clarifying the change in behaviours between $n=5$ and $n=6$.
These points remain open for further research.

\section*{Acknowledgements}
E. Tonello thanks P. Capriotti for helpful discussions.

\bibliographystyle{plainnat}
\bibliography{biblio}

\begin{thebibliography}{19}
\providecommand{\natexlab}[1]{#1}
\providecommand{\url}[1]{\texttt{#1}}
\expandafter\ifx\csname urlstyle\endcsname\relax
  \providecommand{\doi}[1]{doi: #1}\else
  \providecommand{\doi}{doi: \begingroup \urlstyle{rm}\Url}\fi

\bibitem[Biere(2016)]{biere2016splatz}
Armin Biere.
\newblock {S}platz, {L}ingeling, {P}lingeling, {T}reengeling, {YalSAT} entering
  the {SAT} competition 2016.
\newblock \emph{Proceedings of {SAT} Competition}, pages 44--45, 2016.

\bibitem[Dubrova and Teslenko(2011)]{dubrova2011sat}
Elena Dubrova and Maxim Teslenko.
\newblock A {SAT}-based algorithm for finding attractors in synchronous
  {B}oolean networks.
\newblock \emph{IEEE/ACM transactions on computational biology and
  bioinformatics}, 8\penalty0 (5):\penalty0 1393--1399, 2011.

\bibitem[Faur{\'e} and Kaji(2018)]{faure2018circuit}
Adrien Faur{\'e} and Shizuo Kaji.
\newblock A circuit-preserving mapping from multilevel to {B}oolean dynamics.
\newblock \emph{J. Theoret. Biol.}, 440:\penalty0 71--79, 2018.

\bibitem[Lammich(2017)]{lammich2017thegrat}
Peter Lammich.
\newblock The {GRAT} tool chain: Efficient {(UN)SAT} certificate checking with
  formal correctness guarantees.
\newblock In \emph{Theory and Applications of Satisfiability Testing -- SAT
  2017}, pages 457--463. Springer International Publishing, 2017.

\bibitem[Remy et~al.(2008)Remy, Ruet, and Thieffry]{remy2008graphic}
{\'E}lisabeth Remy, Paul Ruet, and Denis Thieffry.
\newblock Graphic requirements for multistability and attractive cycles in a
  {B}oolean dynamical framework.
\newblock \emph{Adv. in Appl. Math.}, 41\penalty0 (3):\penalty0 335--350, 2008.

\bibitem[Richard(2008)]{richard2008extension}
Adrien Richard.
\newblock An extension of a combinatorial fixed point theorem of {S}hih and
  {D}ong.
\newblock \emph{Adv. in Appl. Math.}, 41\penalty0 (4):\penalty0 620--627, 2008.

\bibitem[Richard(2010)]{richard2010negative}
Adrien Richard.
\newblock Negative circuits and sustained oscillations in asynchronous automata
  networks.
\newblock \emph{Adv. in Appl. Math.}, 44\penalty0 (4):\penalty0 378--392, 2010.

\bibitem[Richard(2011)]{richard2011local}
Adrien Richard.
\newblock Local negative circuits and fixed points in non-expansive {B}oolean
  networks.
\newblock \emph{Discrete Appl. Math.}, 159\penalty0 (11):\penalty0 1085--1093,
  2011.

\bibitem[Richard and Comet(2007)]{richard2007necessary}
Adrien Richard and Jean-Paul Comet.
\newblock Necessary conditions for multistationarity in discrete dynamical
  systems.
\newblock \emph{Discrete Appl. Math.}, 155\penalty0 (18):\penalty0 2403--2413,
  2007.

\bibitem[Richard and Ruet(2013)]{richard2013kernels}
Adrien Richard and Paul Ruet.
\newblock From kernels in directed graphs to fixed points and negative cycles
  in {B}oolean networks.
\newblock \emph{Discrete Appl. Math.}, 161\penalty0 (7):\penalty0 1106--1117,
  2013.

\bibitem[Ruet(2016)]{ruet2016local}
Paul Ruet.
\newblock Local cycles and dynamical properties of {B}oolean networks.
\newblock \emph{Math. Structures Comput. Sci.}, 26\penalty0 (04):\penalty0
  702--718, 2016.

\bibitem[Ruet(2017)]{ruet2017negative}
Paul Ruet.
\newblock Negative local feedbacks in {B}oolean networks.
\newblock \emph{Discrete Appl. Math.}, 221:\penalty0 1--17, 2017.

\bibitem[Shih and Dong(2005)]{shih2005combinatorial}
Mau-Hsiang Shih and Jian-Lang Dong.
\newblock A combinatorial analogue of the {J}acobian problem in automata
  networks.
\newblock \emph{Adv. in Appl. Math.}, 34\penalty0 (1):\penalty0 30--46, 2005.

\bibitem[Slepian(1953)]{slepian1953number}
David Slepian.
\newblock On the number of symmetry types of {B}oolean functions of n
  variables.
\newblock \emph{Canad. J. Math}, 5\penalty0 (2):\penalty0 185--193, 1953.

\bibitem[Tamura and Akutsu(2009)]{tamura2009detecting}
Takeyuki Tamura and Tatsuya Akutsu.
\newblock Detecting a singleton attractor in a {B}oolean network utilizing
  {SAT} algorithms.
\newblock \emph{IEICE Transactions on Fundamentals of Electronics,
  Communications and Computer Sciences}, 92\penalty0 (2):\penalty0 493--501,
  2009.

\bibitem[Thomas(1981)]{thomas1981relation}
R.~Thomas.
\newblock On the relation between the logical structure of systems and their
  ability to generate multiple steady states or sustained oscillations.
\newblock In \emph{Numerical Methods in the Study of Critical Phenomena}, pages
  180--193. Springer Berlin Heidelberg, 1981.
\newblock ISBN 978-3-642-81703-8.

\bibitem[Tonello(2017)]{tonello2017ontheconversion}
Elisa Tonello.
\newblock On the conversion of multivalued gene regulatory networks to
  {B}oolean dynamics.
\newblock \emph{arXiv preprint arXiv:1703.06746}, 2017.

\bibitem[Varela et~al.(2018)Varela, Lynce, Manquinho, Chaouiya, and
  Monteiro]{varela2018stable}
Pedro~L Varela, In{\^e}s Lynce, Vasco Manquinho, Claudine Chaouiya, and Pedro~T
  Monteiro.
\newblock Stable states of {B}oolean regulatory networks composed over
  hexagonal grids.
\newblock \emph{Electronic Notes in Theoretical Computer Science},
  335:\penalty0 113--130, 2018.

\bibitem[Wetzler et~al.(2014)Wetzler, Heule, and Hunt]{wetzler2014drat}
Nathan Wetzler, Marijn J.~H. Heule, and Warren~A. Hunt.
\newblock Drat-trim: Efficient checking and trimming using expressive clausal
  proofs.
\newblock In \emph{Theory and Applications of Satisfiability Testing -- SAT
  2014. Lecture Notes in Computer Science}, volume 8561, pages 422--429.
  Springer, 2014.

\end{thebibliography}

\appendix
\section{Boolean networks with antipodal attractive cycles}\label{appendix:antip6}
In the following, we write $e^j$ for the state such that $e^j_i = 0$ for $i \neq j$, and $e^j_j = 1$.
The following definition can be found in~\cite{ruet2016local,ruet2017negative}.
\begin{definition}\label{def:antipodal}
  A cycle is called \emph{antipodal attractive cycle} if it is obtained from the cycle
  \begin{equation}\label{eq:antipodal}
  (\0, e^1, e^1+e^2,\dots, e^1+\cdots+e^n, e^2+\cdots+e^n, \dots, e^n, \0)
  \end{equation}
  by application of a map $\psi_I \circ \sigma$, with $I\subseteq\{1,\dots,n\}$ and $\sigma \in S_n$.
\end{definition}
We describe here a procedure for constructing maps with an antipodal attractive cycle
and no local negative circuits for $n\geq 6$, thus extending
the method described in~\cite{ruet2017negative} to the case $n=6$.

The idea of the construction is the following.
The regulatory graph of the map consisting of the antipodal attractive cycle $\C$, and all other states
fixed, admits many local negative circuits. These circuits belong to graphs
$G_f(x)$ with $x \in \C$, since the regulatory graph at fixed points cannot admit a negative circuit (Lemma~\ref{lemma:sign_circuit_parity}).
By carefully modifying the map around the antipodal cycle, one can eliminate the
local negative circuits, while maintaining the other states fixed.

We start by setting the notation for the states in the antipodal cycle. We set
\[a^i = \sum_{k=1}^{i-1} e^i,\]
\[a^{n+i} = \overline{a^i},\]
for $i = 1, \dots, n$.
Observe that $a^{i+1} = a^i + e^{i}$, and that the antipodal cycle is defined by $(a^1=\0, a^2,\dots,a^n,a^{n+1},\dots,a^{2n},a^1)$.
We extend the notation for the $e^i$ by setting $e^{i+kn} = e^i$ for $i \in \{1,\dots,n\}$, $k \in \mathbb{Z}$.
Then, we define
\begin{equation*}
  \begin{aligned}
    b^i & = a^i + e^{i+1}, \\
    c^i & = a^i + e^{i+1} + e^{i+2} = b^i + e^{i+2}, \\
    d^i & = a^i + e^{i+1} + e^{i+3} = b^i + e^{i+3}, \\
  \end{aligned}
\end{equation*}
for $i=1, \dots, 2n$.
Set $a^{i+2kn}=a^i$ for $i = \{1, \dots, 2n\}$ and $k \in \mathbb{Z}$, and similarly for the states $b^i$, $c^i$ and $d^i$.
We define the map $f$ as follows:
\begin{equation*}
  \begin{aligned}
    f(a^i) & = a^{i+1}, \\
    f(b^i) & = a^{i+2}, \\
    f(c^i) & = a^{i+4}, \\
    f(d^i) & = a^{i+4}, \\
  \end{aligned}
\end{equation*}
for $i=1, \dots, 2n$, while all other states are fixed.

The map $f$ is well defined, and the asynchronous dynamics it defines admits an antipodal attractive cycle,
whereas its regulatory graph admits no local negative circuits.
The proof is similar to the one presented in~\cite{ruet2017negative}, and is omitted.
The map obtained for $n=6$ is represented in \cref{fig:antip6}.

\begin{figure}[t]
  \centering
  \begin{tikzcd}[ampersand replacement=\&,column sep=small,row sep=small]
    \& \& \mybox{101010\\101100\\111010} \arrow[loop,in=-30,out=30,looseness=1.5] \arrow[d,line width=1.1] \arrow[out=180,in=90,ddl] \& \mybox{110101\\110110\\111101} \arrow[loop,in=-30,out=30,looseness=1.5] \arrow[out=180,in=+20,dl] \arrow[d,line width=1.1] \& \mybox{011010\\111011\\011110} \arrow[loop,in=-30,out=30,looseness=1.5] \arrow[d,line width=1.1] \arrow[out=180,in=+20,dl] \& \& \\
    \& \& 111110 \arrow[r,line width=1.1] \& 111111 \arrow[r,line width=1.1] \& 011111 \arrow[dr,bend left=20,line width=1.1] \& \& \\
     \mybox{010100\\110100\\011000} \arrow[r,line width=1.1] \arrow[dr] \arrow[loop,in=180,out=150,looseness=2] \& 111100 \arrow[ur,bend left=20,line width=1.1] \& \&\& \& 001111 \arrow[d,line width=1.1] \&\mybox{011101\\101111\\101101} \arrow[l,line width=1.1] \arrow[llu,in=0,out=150] \arrow[loop,in=0,out=-30,looseness=2] \\
    \mybox{101001\\101000\\110001} \arrow[r,line width=1.1] \arrow[dr] \arrow[loop,in=180,out=150,looseness=2] \& 111000 \arrow[u,line width=1.1] \& \&\& \& 000111 \arrow[d,line width=1.1] \& \mybox{001110\\010111\\010110} \arrow[l,line width=1.1] \arrow[ul] \arrow[loop,in=0,out=-30,looseness=2] \\
    \mybox{010010\\010000\\100010} \arrow[r,line width=1.1] \arrow[rrd,out=-30,in=180] \arrow[loop,in=180,out=150,looseness=2] \& 110000 \arrow[u,line width=1.1] \& \&\& \& 000011 \arrow[dl,bend left=20,line width=1.1] \& \mybox{100111\\001011\\101011} \arrow[l,line width=1.1] \arrow[ul] \arrow[loop,in=0,out=-30,looseness=2] \\
    \& \& 100000 \arrow[ul,bend left=20,line width=1.1] \& 000000 \arrow[l,line width=1.1] \& 000001 \arrow[l,line width=1.1] \& \& \\
    \& \& \mybox{100001\\000100\\100101} \arrow[loop,out=210,in=150,looseness=1.5] \arrow[u,line width=1.1] \arrow[out=0,in=200,ur] \& \mybox{000010\\001001\\001010} \arrow[loop,out=210,in=150,looseness=1.5] \arrow[u,line width=1.1] \arrow[out=0,ur,in=200] \& \mybox{000101\\010011\\010101} \arrow[loop,out=210,in=150,looseness=1.5] \arrow[u,line width=1.1] \arrow[out=0,uur,in=270] \& \& \\
  \end{tikzcd}
  \caption{Dynamics for a regulatory network with an antipodal attractive cycle and admitting no local negative circuits, for $n=6$.
           The fixed points are omitted.
           The synchronous dynamics coincides for the states in the same box, and is represented with bold arrows.
           The additional edges are asynchronous.
           }\label{fig:antip6}
\end{figure}
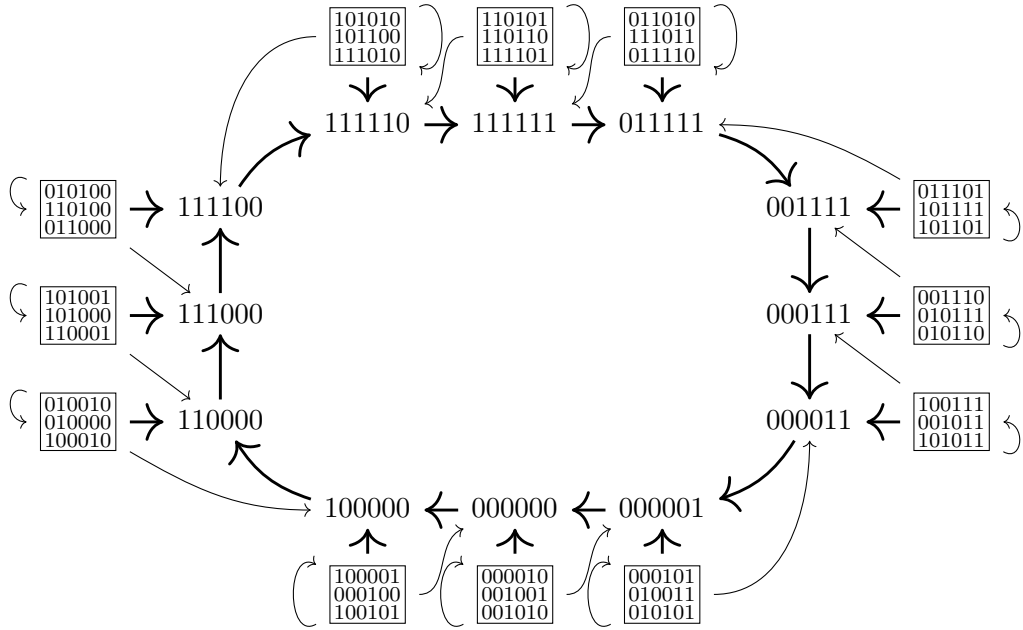

\end{document}